\documentclass[reqno,a4paper,11pt]{amsart}

\usepackage{times} 
\usepackage{amsmath,amssymb,amsfonts,mathrsfs}
\usepackage{a4wide}
\usepackage{booktabs}
\usepackage{url}
\usepackage{hyperref}
\usepackage[english]{babel} 
\usepackage{amsthm}

\theoremstyle{plain}
\newtheorem{thm}{Theorem}

\newtheorem{proposition}[thm]{Proposition}
\newtheorem{corollary}[thm]{Corollary}
\newtheorem{lemma}[thm]{Lemma}

\theoremstyle{definition}
\newtheorem{definition}{Definition}

\theoremstyle{remark}
\newtheorem{remark}{Remark}

\renewcommand{\leq}{\leqslant}
\renewcommand{\le}{\leqslant}

\renewcommand{\ge}{\geqslant}

\newcommand{\eqdef}{\stackrel{\text{def}}{=}}
\newcommand{\F}{\mathbb{F}}
\newcommand{\K}{\mathbb{K}}
\newcommand{\fq}{\F_{q}}
\newcommand{\fqm}{\F_{q^m}}

\newcommand{\word}[1]{\boldsymbol{#1}}
\newcommand{\bv}{\word{b}}
\newcommand{\cv}{\word{c}}
\newcommand{\ev}{\word{e}}
\newcommand{\gv}{\word{g}}
\newcommand{\hv}{\word{h}}
\newcommand{\mv}{\word{m}}
\newcommand{\vv}{\word{v}}
\newcommand{\xv}{\word{x}}

\newcommand{\zv}{\word{z}}

\newcommand{\mat}[1]{\boldsymbol{#1}}
\newcommand{\Am}{\mat{A}}
\newcommand{\Bm}{\mat{B}}
\newcommand{\Cm}{\mat{C}}
\newcommand{\Dm}{\mat{D}}

\newcommand{\Gm}{\mat{G}}
\renewcommand{\Im}{\mat{I}}

\newcommand{\Lm}{\mat{L}}
\newcommand{\Mm}{\mat{M}}
\newcommand{\Pm}{\mat{P}}
\newcommand{\Qm}{\mat{Q}}
\newcommand{\Rm}{\mat{R}}
\newcommand{\Sm}{\mat{S}}
\newcommand{\Tm}{\mat{T}}
\newcommand{\Xm}{\mat{X}}
\newcommand{\Ym}{\mat{Y}}
\newcommand{\Zm}{\mat{Z}}

\newcommand{\ZZ}{\mat{0}}
\newcommand{\Gp}{\mat{G}_{\rm pub}}
\newcommand{\Cp}{\code{C}_{\rm pub}}
\newcommand{\tp}{t_{\rm pub}}

\newcommand{\gab}[2]{\CG_{#1}\left(#2\right)}
\newcommand{\code}[1]{\mathscr{#1}}
\newcommand{\dual}[1]{{#1}^\bot}
\newcommand{\CA}{\code{A}}
\newcommand{\CC}{\code{C}}
\newcommand{\CG}{\code{G}}

\newcommand{\norm}[1]{\left | #1  \right |}
\newcommand{\rank}{{\normalfont \texttt{rank}}}
\newcommand{\GL}{{\normalfont \textsf{GL}}}
\newcommand{\distor}{\mathcal{D}}

\newcommand{\MS}[3]{\mathcal{M}_{#1,#2}\left(#3\right)}

\title{Improved Cryptanalysis of Rank Metric Schemes Based on Gabidulin Codes}

\author{Ayoub Otmani}	
	\address{University of Rouen,  LITIS}
	\email{ayoub.otmani@univ-rouen.fr} 
	
\author{Herv\'e Tal\'e Kalachi}
\address{University of Rouen, LITIS 
 		\&
		University of Yaounde 1, Department of Mathematics, ERAL, Cameroon}
	\email{hervekalachi@gmail.com}
\author{S\'elestin Ndjeya}
 	\address{
	University of Yaounde 1, Department of Mathematics, ERAL, Cameroon.}
	\email{ndjeyas@yahoo.fr}

\begin{document}

\begin{abstract}
	We prove that any variant of the GPT cryptosystem which uses a right column scrambler over the extension field
as advocated by the works of Gabidulin \textit{et al.}  with the goal to resist to Overbeck's structural attack are actually still 
vulnerable to that attack. We show that by applying  the Frobenius operator appropriately on the public key, it is possible to build a Gabidulin code having the same dimension as the original secret Gabidulin code 
but with a lower length. In particular, the code obtained by this way correct less errors than the secret one 
but its error correction capabilities are beyond the number of errors added by a sender, and consequently
an attacker is able to decrypt any ciphertext with this degraded Gabidulin code.
We also considered the case where an isometric transformation is applied  in conjunction  with a right column scrambler which has its entries in the extension field. We proved that this protection is useless both in terms of performance and security.
Consequently, our results show that all the existing techniques aiming to hide the inherent algebraic structure of Gabidulin codes have failed.

\end{abstract}

\maketitle


\section{Introduction}

The emergence of the post-quantum cryptography was mainly enabled thanks to  Shor 
who proved \cite{S94a,S97} that the discrete log problem and the factorization 
can be solved in polynomial time with an hypothetical quantum computer. Recent progress in solving the discrete log problem, 
and the fact that important industrial investments are made to build a quantum computer 
have aroused concerns about the foundations of the real-world cryptography, prompting people to seek serious 
post-quantum alternatives.

Among all the existing  solutions,
McEliece scheme \cite{M78} is one of the oldest post-quantum public key encryption scheme. The innovative McEliece's approach 
rests on the use of the theory of error-correcting codes to design a one-way function of the form $\mv \mapsto \mv \Gm + \ev$
where $\Gm$ generates a vector subspace of $\F_2^n$ and $\ev$ is a random binary error vector of Hamming weight $\tp$. McEliece  
used binary Goppa codes which are well-known for having a very fast decoding algorithm.
Designed in 1978, it has withstood several attack attempt but it suffers from an important drawback due to the enormous size of the
public keys. 
In order to solve this problem, several modifications of the
scheme have been proposed among which the use of rank metric
codes instead of the Hamming metric. The first rank-metric scheme was proposed in \cite{GPT91} by Gabidulin,
Paramonov and Tretjakov which is now called the GPT cryptosystem.
This scheme  can be seen as an analogue of the McEliece scheme public key cryptosystem based on the class of Gabidulin codes. 
An important operation in the key generation of the GPT cryptosystem is the ``hiding'' phase where the secret generator matrix 
$\Gm$ undergoes a  transformation to mask the inherent algebraic structure of the associated Gabidulin code. 
This transformation  is  a 
probabilistic algorithm  that adds some randomness to its input.
Originally, the authors  in \cite{GPT91} proposed to use  a \emph{distortion} 
transformation  that takes $\Gm$ and outputs the public matrix  
$\Gp = \Sm(\Gm + \Xm)$ where $\Xm$ is a random matrix with a prescribed rank 
$t_{\Xm}$ and $\Sm$ is an invertible matrix. The presence of a  distortion matrix has however an impact: 
the sender has to  add an error vector whose rank weight is $\tp = t - t_{\Xm}$ where $t$ is the error correction capability of the secret 
underlying Gabidulin code. Hence, roughly speaking,  the hiding phase publishes a degraded code in terms of 
error correction.

Gabidulin codes are often seen as  equivalent of Reed-Solomon codes in the Hamming metric and like them, 
they are highly structured. That is the reason why their use in the GPT cryptosystem has been the subject to
several attacks. Gibson was the first to prove
the weakness of the system through a series of successful attacks \cite{G95,G96}.
Following this failures, the first works which modified the GPT scheme to avoid Gibson's attack were published in \cite{GO01,GOHA03}.  
The idea is to hide further the structure of Gabidulin code by considering  isometries for the rank metric.
Consequently, a \emph{right column scrambler} $\Pm$ is  introduced which is an invertible matrix with its entries in the base field $\fq$
while the ambient space of the Gabidlun code is $\fqm^n$. 
But Overbeck designed in \cite{O05,O05a,O08}  a more
general attack that dismantled all the existing  modified GPT cryptosystems. 
His approach consists in applying an operator $\Lambda_i$ which applies $i$ times  the Frobenius operation 
on the public generator matrix $\Gp$. Overbeck observed that the dimension increases by $1$ each time the Frobenius is applied. He then
proved that by taking $i = n - k - 1$ the codimension becomes $1$ if $k$ is the rank of $\Gp$ (which also  the dimension of the associated Gabidulin code).
This phenomenon is a clearly distinguishing property of a Gabidulin code which cannot be encountered for instance with a random linear code where the dimension would increase by $k$ for each use of the Frobenius operator.

Overbeck's attack uses crucially  two important facts, namely the column scrambler matrix $\Pm$ is defined on the 
based field $\fq$ and the codimension of $\Lambda _{n-k-1}\left(\Gp \right)$ is equal to $1$.
Several works then proposed to resist to this attack  either by taking special 
distortion matrix so that the second property is not true as in \cite{L10,RGH10}, or 
by taking a column scrambler matrix defined over the extension field $\fqm$ as in \cite{G08,GRH09,RGH11}. 

\medskip

In this paper, we  study  the security of the second approach. We show that even if the column scrambler is defined 
on the extension field as it is done in  \cite{G08,GRH09,RGH11},  it is still possible to recover a secret Gabidulin code using precisely Overbeck's technique. Our analysis shows that by applying the operator $\Lambda_i$  with $i < n - k - 1$, we obtain a  Gabidulin code  whose error correction $t^*$ is indeed strictly less than the error correction of the secret original Gabidulin code
but, $t^*$ is strictly greater than the number of added errors  $\tp$. In other words, an attacker is still able to decrypt any ciphertext 
and consequently, all the scheme presented in  \cite{G08,GRH09,RGH11} are actually not resistant to Overbeck's attack 
unlike what it was claimed by the authors. When the attack is implemented with the recommended parameters of   
\cite{GRH09,RGH11}, our experimental results show that the attack is very fast (less than one second).
In particular, our results outperform those given
in \cite{GRS16,HMR16} which were for a while the best attacks against the schemes of \cite{G08,GRH09,RGH11}.
Note that the attacks of \cite{GRS16,HMR16}  are generic decoding algorithms  whereas  our approach is directed towards recovering the structure of a Gabidulin code.

\section{Preliminary Notion}

The finite field with $q$ elements is denoted by $\F_q$ where $q$ is a power of a prime number.
For any subfield $\K \subseteq \F$ of a field $\F$
and for any positive integers  $k$ and $n$ such that $k \le n$,
the $\K$-vector space spanned by  $\bv_1,\dots{},\bv_k$  where each $\bv_i \in \F^n$ is denoted by 
$\sum_{i=1}^k \K \; \bv_i$.
The set of matrices with $m$ rows and $n$ columns 
and entries in $\F$ is denoted by $\MS{m}{n}{\F}$. 
The group of invertible matrices of size $n$ over $\F$ is denoted by $\GL_n(\F)$.

\begin{definition}[Rank weight]
Let $\Am$ be a matrix from $\MS{m}{n}{\F}$ where $m$ and $n$ are positive integers. The \emph{rank weight} 
of $\Am$ denoted by $\norm{\Am}$ is the rank of $\Am$. The \emph{rank distance} between two matrices 
$\Am$ and $\Bm$ from $\MS{m}{n}{\F}$ is defined as $\norm{\Am - \Bm}$.
\end{definition}

It is a well-known fact that the rank distance on $\MS{m}{n}{\F}$ has the properties of a metric. 
But in the context of the rank-metric cryptography, this rank distance is rather defined for vectors $\xv \in \fqm^n$. 
The idea is to consider the field $\fqm$ as an
$\fq$-vector space and hence any vector  $\xv \in \fqm^n$ as a matrix from $\MS{m}{n}{\fq}$ by decomposing 
each entry $x_i \in \fqm$ into an $m$-tuple of $\fq^m$ with respect to an arbitrary basis of $\fqm$.  The rank weight 
of $\xv$ also denoted by $\norm{\xv}$ is then its rank\footnote{This rank is of course independent of the choice of the basis  of $\fqm$ since the rank of a matrix is invariant when multiplied by an invertible matrix.} 
viewed as a matrix of $\MS{m}{n}{\fq}$. Hence, it is possible to define a new metric on $\fqm^n$ that we recall explicitly 
in the following.

\begin{definition}
Let us consider the finite field extension $\fqm/\fq$ of degree $m \ge 1$.
The \emph{rank weight} of a vector $\xv = \left(x_{1},x_{2},...,x_{n}\right)$ in $\fqm^n$  denoted by $\norm \xv$
is  the dimension of the $\fq$-vector space generated by $\{x_1,\dots{},x_n\}$
\begin{equation}
\norm \xv = \dim \sum_{i=1}^n \fq x_i.
\end{equation}
Similarly,  the  column rank over $\fq$ for any matrix $\Mm$ from $\MS{k}{n}{\fqm}$
is also denoted by $\norm \Mm$. 
\end{definition}
\begin{remark}
Note again that  $\norm \Mm$ represents $\dim \sum_{i}^n \fq \Mm_{i}$ where $\Mm_1,\dots{},\Mm_n$ 
are the columns of $\Mm$, that is to say the maximum number of columns that are linearly independent over $\fq$ 
when each entry is written as an $m$-tuple of $\fq^m$ with respect to a basis of $\fqm$.  
Equivalently,  $\norm{\Mm}$  is also the maximal number of rows that are linearly independent over $\fq$.
\end{remark}

\begin{proposition} \label{prop:rank_reduction}
Let $\Mm$ be  a matrix from $\MS{k}{n}{\fqm}$ and set  $s = \norm{\Mm}$ with $s  < n$. 
There exist then $\Mm^*$ in $\MS{k}{s}{\fqm}$ with $\norm{\Mm^*} = s$ and
$\Tm$ in $\GL_n(\fq)$ such that:
\begin{equation}
\Mm \Tm= (\Mm^* \mid \ZZ)
\end{equation}
In particular for any $\xv \in \fqm^{n}$ such that  $\norm{\xv} = s$ 
there exists  $\Tm$ in $\GL_n(\fq)$  for which $\xv \Tm=(\xv^* \mid \ZZ)$ where 
$\xv^* \in \fqm^{s}$ and $\norm{\xv^*}=s$.
\end{proposition}
This permits to state the following corollary.

\begin{corollary} \label{cor:uprank}
For any $\Mm \in \MS{k}{n}{\fqm}$ and for any  $\mv \in \fqm^k$
\begin{equation}
\norm{\mv\Mm} \leq \norm{\Mm}
\end{equation}
\end{corollary}

\begin{definition}
For any $x$ in $\fqm$ and for any integer $i$, the quantity $x^{q^{i}}$ is denoted by $x^{[i]}$. This notation is extended to vectors  $\xv^{[i]} = (x_{1}^{[i]},\dots{},x_{n}^{[i]})$ and matrices $\Mm^{[i]}= \left ( m_{ij}^{[i]} \right)$.
\end{definition}

We also give two lemmas that will be useful in the sequel.

\begin{lemma} \label{lem:frob_prop}
For any $\Am \in \MS{\ell}{s}{\fqm}$ and  $\Bm \in \MS{s}{n}{\fqm}$, and for any $\alpha$ and $\beta$ in $\fq$:
\[
\left( \alpha \Am + \beta \Bm \right)^{\left[i\right]} = \alpha \Am^{\left[i\right]} + \beta \Bm^{\left[i\right]}
~~~\text{ and }~~~
\left(\Am \Bm \right)^{\left[i\right]}=\Am^{\left[i\right]} \Bm^{\left[i\right]}.
\]
In particular if $\Sm$ is in $\GL_n(\fqm)$ then $\Sm^{\left[i\right]}$ also belongs to $\GL_n(\fqm)$.
\end{lemma}

\begin{lemma}\label{InvBlockMat}
Let $\Pm=%
\left(
\begin{matrix}
\Am & \ZZ \\
\Cm & \Dm%
\end{matrix}
\right)%
$ where $\Am$ and $\Dm$ are square matrices. Then $\Pm$ is non singular if and only if $\Am$ and $\Dm$ are non singular and  
the inverse of $\Pm$ is:
$$
\Pm^{-1}=%
\left(
\begin{matrix}
\Am^{-1} & \ZZ \\
-\Dm^{-1}\Cm \Am^{-1} & \Dm^{-1}%
\end{matrix}
\right)%
$$
\end{lemma}

Let us recall that a (linear) code of length $n$ over a finite field $\F$ is a linear subspace of $\F^n$. Elements of a code are called \emph{codeword}. A matrix whose rows form a basis of a code is called a \emph{generator matrix}.
The dual of a code $\CC\subset \F^n$ is the linear space denoted by $\dual{\CC}$ containing vectors $\zv \in \F^n$ such that:
\[
\forall \cv \in \CC, \;\; \sum_{i=1}^n c_i z_i = 0. 
\]
An algorithm $D$ is said to 
decode $t$ errors in a code $\CC \subset \F^n$ if for any $\cv \in \CC$ and for any $\ev \in \F^n$ such that $\norm \ev \le t$ we have $D(\cv + \ev) = \cv$. 
Generally, we call such a vector $\ev$ an \emph{error} vector.
We introduce now an important family of codes known for having an efficient decoding algorithm.

\begin{definition}[Gabidulin code]
Let $\gv \in \fqm^{n}$ such that $\norm{\gv}=n$.
The $(n,k)-$Gabidulin code denoted by $\gab{k}{\gv}$ is the code with a generator matrix $\Gm$ where: 
\begin{equation} \label{gab:genmat}
\Gm=
\begin{pmatrix}
g_{1}^{[0]} & \cdots{} & g_{n}^{[0]} \\
\vdots{}      &              &  \vdots{} \\
g_{1}^{[k-1]} & \cdots{} & g_{n}^{[k-1]}
\end{pmatrix}.
\end{equation}
\end{definition}
Gabidulin codes are known to possess a fast decoding algorithm that 
can decode errors of weight $t$ provided that $t \leq \lfloor \frac{1}{2}(n-k) \rfloor$.
Furthermore the dual of a Gabidulin code $\gab{k}{\gv}$ is also a Gabidulin code.

We end this section by an important well-known property about Gabidulin codes.

\begin{proposition}\label{ScramblingAGabidulinAtRight}
Let $\gab{k}{\gv}$ be a Gabidulin code of length n with generator
matrix $\Gm$ and $\Tm \in \GL_n(\fq)$. 
Then $\Gm\Tm$ is a generator matrix of the Gabidulin code $\gab{k}{\gv \Tm}$
\end{proposition}

\begin{proof}

From Lemma~\ref{lem:frob_prop}, we have $\left(\gv \Tm \right)^{\left[ i \right]}=\gv^{{\left[ i \right]}}\Tm$. \qed
\end{proof}

\section{Rank Metric Encryption Schemes}

The concept of rank metric cryptography appeared in \cite{GPT91} where the authors propose a 
public key encryption scheme using codes in a rank metric framework. They  
adapted McEliece's general idea \cite{M78} developed for the Hamming metric to the rank metric context. The key tool in the design is to focus on linear codes  having a fast rank-metric decoding algorithm like Gabidulin codes.
In this section, we recall the general principle that underlies all the existing rank encryption metric schemes.

During the key generation phase, the integers $k$, $\ell$, $n$ and 
$m$ are chosen such that $k<n\leq m$ and $0 \le \ell \ll n$. It then 
randomly picks  $\gv \in \fqm^n$ with $\norm{\gv} = n$ and
defines  $\Gm \in \MS{k}{n}{\fqm}$ as in \eqref{gab:genmat}, that is to say $\Gm$ is
a generator matrix of the Gabidulin code $\gab{k}{\gv}$. The error-correcting capacity of  
$\gab{k}{\gv}$ is denoted by $t \eqdef \lfloor \frac{1}{2} (n - k) \rfloor$.
An important step in the key generation is the ``hiding'' phase where $\Gm$ undergoes a 
transformation to mask the algebraic structure of Gabidulin codes. This transformation  is  actually a 
probabilistic algorithm  that adds some randomness to its input.
Originally, the authors  in \cite{GPT91} proposed to use  a \emph{distortion} 
transformation $\distor  : \fqm^{k \times n} \longrightarrow \fqm^{k \times n}$ that 
sends any $\Gm$ to $\distor(\Gm) = \Sm(\Gm + \Xm)$ where $\Xm$ is a random matrix from 
$\fqm^{k\times n}$ with a prescribed rank 
$t_{\Xm}$ and $\Sm$ is an invertible matrix.  The public key is then 
$\Gp = \distor(\Gm)$ with the parameter $\tp = t - t_{\Xm}$ while the private key is $(\Sm, \Gm)$.
The  encryption  algorithm takes as input a plaintext $\mv \in \fqm^{k}$ and generates a random
$\ev \in \fqm^{n}$ such that $\norm{\ev} \leq \tp$ in order to compute the
ciphertext  $\cv=\mv \Gp + \ev$.
In the  decryption step the decoding algorithm of the Gabidulin code  $\gab{k}{\gv}$ is applied to the 
ciphertext.  This word can be decoded since the underlying codeword is corrupted by the error vector 
$\mv \Sm \Xm +\ev$ whose rank weight is $\norm{\mv \Sm \Xm +\ev} \le \norm{\mv \Sm \Xm} + \norm{\ev} \le t$ since  
by Corollary~\ref{cor:uprank} we have $ \norm{\mv \Sm \Xm} \leq t_{\Xm}$.

However, Gibson proved \cite{G95,G96} that the GPT encryption scheme \cite{GPT91}  is vulnerable 
to a polynomial time key recovery attack.
Consequently, Gabidulin and Ourivski proposed  in \cite{GO01} a reparation by considering a more general 
hiding transformation combining a distortion matrix $\Xm$ and a right column scrambler $\Pm$. The hidden generator 
matrix is 
more precisely of the form:
\begin{equation} \label{GPTX1X2}
\distor(\Gm) = \Sm \left(\Xm_{1} \mid \Gm + \Xm_{2}\right) \Pm
\end{equation}
where $\Xm_{1} \in \MS{k}{\ell}{\fqm}$, $\Xm_{2} \in \MS{k}{n}{\fqm}$ such that 
$\norm{\Xm_{2}}<t$  and  $\Pm \in \GL_{n+\ell}(\fq)$. 
The public generator matrix is again $\Gp \eqdef \distor(\Gm)$ which constitutes 
the public key with the public parameter $\tp \eqdef t - t_2$ where $t_2 \eqdef \norm{\Xm_{2}}$.
The decryption computes $\Pm^{-1} =  ( \Qm_1 \mid  \Qm_2)$ where $\Qm_1 \in \MS{(n+\ell)}{\ell}{\fq}$ and
$\Qm_2 \in \MS{(n+\ell)}{n}{\fq}$.
The last $n$ components of $\cv \Pm^{-1}$ is the vector
$\mv \Sm \Gm+\mv\Sm\Xm_{2}+ \ev\Qm_2$ and
since $\norm{\ev \Qm_2} \leq \norm{\ev}$ and  $\norm{\mv \Sm \Xm_{2}} \leq \norm{\Xm_{2}}$, it follows that $%
\norm{\mv \Sm \Xm_{2} + \ev \Qm_2} \leq t$.
Applying a fast decoding algorithm to
the last $n$ components of $\cv \Pm^{-1}$ 
 allows the legitimate user to get $\mv \Sm$
and easily $\mv$.

We now state our first result about Gabidulin and Ouriviski reparation which proves 
that we can always  consider $\Xm_2 = 0$.
 
\begin{proposition}\label{prop:XtoX2=0}
Let $\Gp$ be as in \eqref{GPTX1X2} and assume that $\norm{\Xm_{2}} = t_2$.
There exist $\Pm^* \in \GL_{n+\ell}(\fq)$, $\Xm^* \in \MS{k}{(\ell + t_2)}{\fqm}$ 
and a matrix $\Gm^*$ that generates an $(n - t_2)-$Gabidulin code  $\gab{k}{\gv^*}$
such that
\begin{equation} \label{eq:transf}
\Gp = \Sm \left (\Xm^*  \mid \Gm^* \right) \Pm^*.
\end{equation}
Furthermore, the error correction capability $t^*$ of $\gab{k}{\gv^*}$  is equal to $t - \frac{1}{2}t_2$, and hence
$t^* > \tp$.
\end{proposition}
\begin{proof}
Since $\norm{\Xm_2} = t_2$ then by Proposition~\ref{prop:rank_reduction} there exist  $\Tm_2$ in $\GL_{n}(\fq)$ and $\Xm'_2$ in $\MS{k}{t_2}{\fqm}$ such that $
\Xm_2 \Tm_2 =\left( \Xm'_2 \mid \ZZ \right)$. 
So by letting $\Tm =
\begin{pmatrix}
\Im_\ell & \ZZ \\
\ZZ & \Tm_2
\end{pmatrix}$
we then have:
\begin{eqnarray*}
\Gp  = \Sm \left(\Xm_{1} \mid \Gm + \Xm_{2}\right) \Pm 
        &= & \Sm \left(\Xm_{1} \mid \Gm \Tm_2 + \Xm_{2} \Tm_2 \right) \Tm^{-1} \Pm \\
      &= & \Sm \left(\Xm_{1} \mid \Gm' + \Xm_{2} \Tm_2 \right) \Qm   
\end{eqnarray*}
where  $\Gm' = \Gm \Tm_2$ and $\Qm =  \Tm^{-1} \Pm $. Note that 
$\Gm'$ generates the $(n,k)-$Gabidulin code $\gab{k}{\gv'}$ with $\gv' = \gv \Tm_2 = (g'_1,\dots{},g'_n)$.
Let us decompose $\Gm'$ as $(\Gm'_1 \mid \Gm'_2)$
where $\Gm'_1 \in \MS{k}{t_2}{\fqm}$  and  $\Gm'_2 \in \MS{k}{(n - t_2)}{\fqm}$ we then have:
$$
\Gm' + \Xm_{2} \Tm_2 
=  \left( \Gm'_1 + \Xm'_{2} \mid \Gm'_2 \right ) 
$$
By setting $\Xm = \left(\Xm_1 \mid \Gm'_1 + \Xm'_2 \right)$ we get \eqref{eq:transf} and $\Gm'_2$ 
generates the $(n-t_2,k)-$Gabidulin $\gab{k}{\gv'_2}$ where $\gv'_2 = (g'_{t_2+1},\dots{},g'_n)$. 
The error-correction capability $t^*$ of $\gab{k}{\gv'_2}$   is given by 
$t^* = \frac{1}{2} ( n-t_2-k) =  t - \frac{1}{2}t_2$ which implies $t^* > t - t_2$.
\qed
\end{proof}

The first important consequence of Proposition~\ref{prop:XtoX2=0} is the possibility for a cryptanalyst who is able to derive $(\Sm, \Gm^*, \Pm^*)$ from $\Gp$  so that \eqref{eq:transf} is satisfied
to decipher any ciphertext $\cv=\mv \Gp + \ev$ with $\norm{\ev} \leq \tp$. 
Thus any successful structural attack on the description  \eqref{eq:transf} leads to
a successful attack on \eqref{GPTX1X2} and conversely since  \eqref{eq:transf}
corresponds to the special case  where $\Xm_2= \ZZ$. 
Therefore the security of the scheme given \cite{GO01} is equivalent to the one of a scheme where $\Xm_2= \ZZ$.

\section{Distinguishing Properties of Gabidulin Codes}

We recall important algebraic properties about Gabidulin codes. It will 
explain why many attacks occur when the underlying code is a Gabidulin one. One key property is 
that Gabidulin codes can be easily distinguished from random linear codes. This singular behavior has been presicely exploited by Overbeck \cite{O05,O05a,O08} to mount attacks.

\begin{definition}
For any integer $i \ge 0$ let
$\Lambda _{i} :  \MS{k}{n}{\fqm}  \longrightarrow  \MS{i k }{n}{\fqm}$ be the $\fq$-linear operator that
maps any $\Mm$ from $\MS{k}{n}{\fqm}$ to $\Lambda_i(\Mm)$ where by definition:
\begin{equation}
\Lambda_i(\Mm) \eqdef 
\begin{pmatrix}
\Mm^{[0]} \\
\vdots{} \\
\Mm^{[i]}
\end{pmatrix}.
\end{equation}
For any code $\code{G}$ generated by a matrix $\Gm$ we denote by $\Lambda_i(\code{G})$ the code generated by  $\Lambda_i(\Gm)$.
\end{definition}

\begin{proposition} \label{prop:dsg_gab}
Let $\gv$  be in $\fqm^n$ with $\norm{\gv} = n$ with $n \leq m$. For any integers $k$ and $i$ such that $k \leq n$ and
$i\leq n-k-1$ we have:
\begin{equation}
\Lambda _{i}\big(\gab{k}{\gv}\big) = \gab{k+i}{\gv}.
\end{equation}
\end{proposition}

The importance of $\Lambda _{i}$ becomes clear when one compares the dimension
of the code spanned by $\Lambda _{i}(\Gm)$ for a randomly drawn matrix $\Gm$ and  the dimension obtained when
$\Gm$ generates a Gabidulin code.

\begin{proposition}
if $\CA \subset \fqm^n$ is a code generated by a random matrix from $\MS{k}{n}{\fqm}$ then with a high probability:
\begin{equation}
\dim \Lambda _{i}(\CA) = 
\min\big \{ n,(i+1)k\big\}
\end{equation}
\end{proposition}

In the case of a Gabidulin code, we get a different  situation as explained by Proposition~\ref{prop:dsg_gab}.
Thus there is property  that is computable in polynomial time distinguishes  a Gabidulin code 
from a random one. This can be used in a cryptanalysis context. In fact, Overbeck \cite{O08} 
has proven that, for a public matrix $\Gm_p$ given by equation \eqref{GPTX1X2} with $\Xm_2 = \ZZ$ 
(in particular all the entries of $\Pm$ belong to $\fq$), 
it is possible (under certain conditions) to find in polynomial time an alternative decomposition of 
$\Gm_p$ of the from $\Sm^* \left( \Xm^* \mid \Gm^* \right) \Pm^* $ using the operator $\Lambda_i$. 
This decomposition allows to decrypt any ciphertext computed with $\Gm_p$. 
The reader can refer to Appendix~\ref{appendix:overb_attack} for details concerning attack.  
The key reason explaining its success is given by the following proposition.

\begin{proposition}
Let us consider $\ell$, $k$ and $n$ be positive integers with $\ell < n$ and $1 \leq k < n$. Let
$\Gm$ be in $\MS{k}{n}{\fqm}$ as a generator matrix of a Gabidulin code, and $\Xm$ 
be a randomly drawn matrix from $\MS{k}{\ell}{\fqm}$. Denote $\CA$ as the
code defined by the generator matrix $\left( \Xm \mid \Gm\right) $. 
Then for any integer $i \ge 0$
\begin{equation} \label{ineq:dsg}
k+i ~ \leq ~ \dim \Lambda _{i}\left(\CA \right) ~ \leq ~ k + i + d
\end{equation}
where $d = \min \Big \{ (i+1)k, \ell \Big \}$.
\end{proposition}

Note that by construction $\ell \le n$ and in Overbeck's attack, the integer  $i$ is equal to $n-k - 1$ so that we have both 
$d = \ell$ and, with high probability, the upper bound in \eqref{ineq:dsg} is  actually an equality, namely 
\[
\dim \Lambda _{n-k-1}\left(\CA \right) = k + (n - k - 1) + d = n + \ell - 1.
\]
This implies that the dimension of $\dual{\Lambda _{i}\left(\CA \right)}$ is equal to $1$. This fact is 
then harnessed in \cite{O08} to recover an equivalent Gabidulin code which enables to decrypt any ciphertext. 

\begin{proposition}[\cite{O08}]
Assume that the public key is $\Gp = \Sm \left(\Xm \mid \Gm \right) \Pm$
with $\Xm \in \MS{k}{\ell}{\fqm}$, $\Pm \in \GL_{n+\ell}(\fq)$ and $\Gm$ generates an $(n,k)-$Gabdidulin code.
If $dim \Lambda _{n-k-1}\left(\CA \right)$ is equal to $1$ then it is possible to recover
with $O\left( (n+\ell)^3 \right)$ field operations alternative matrices $\Xm^* \in \MS{k}{\ell}{\fqm}$, $\Pm^* \in \GL_{n+\ell}(\fq)$ and $\Gm^*$ which generate an $(n,k)-$Gabdidulin code such that 
 \[
 \Gp = \Sm \left(\Xm^* \mid \Gm^* \right) \Pm^*
 \]
\end{proposition}

Overbeck's attack uses crucially  two important facts (see
 Appendix~\ref{appendix:overb_attack} for more details): 
 the column scrambler matrix $\Pm$ is defined on the based field $\fq$, 
 and the codimension of $\Lambda _{n-k-1}\left(\CA \right)$ is $1$.
Several works propose to resist to Overbeck's attack  either by taking special 
distortion matrix so that the second property is not true as in \cite{L10,RGH10}, or 
by taking a column scrambler matrix defined over the extension field $\fqm$ as in 
\cite{G08,GRH09,RGH11}. In this paper, we  solely concentrate on the second approach. 
In \cite{GRS16,HMR16}  new generic 
decoding algorithms are presented whereas  our approach is directed towards 
recovering the structure of a Gabidulin code. We will prove that 
all the existing schemes \cite{G08,GRH09,RGH11} can be broken simply with the techniques developed in \cite{O08}. 

\section{Gabidulin's General Reparation} \label{sec:gab_variant}

In this section, we focus on the reparation given in \cite{G08}. This paper is the first to consider 
a column scrambler matrix defined over the extension field. We describe only the key generation and decryption steps of the scheme since the encryption operation is not modified. 
To the best of our knowledge, no structural attack has been mounted against this description. 
The author claimed that Overbeck's attack is not applicable. But in Proposition~\ref{prop:crypt:rep},
we prove that it is still possible to find an alternative private
key  using precisely Overbeck's technique.

\paragraph{\bf Key generation.}

\begin{enumerate}
\item Pick at random $\gv$ from $\fqm^n$ such that $\norm{\gv} = n$ and
let $\Gm$  be a generator matrix of the Gabidulin code $\gab{k}{\gv}$.

\item Pick at random $\Xm \in \MS{k}{\ell}{\fqm}$, $\Sm$ in $\GL_k(\fqm)$
and $\Pm$ in $\GL_{n+\ell}(\fqm)$ such that
there exist $ \Qm_{11}$ in $\MS{\ell}{\ell}{\fqm}$, $\Qm_{21}$ in $\MS{n}{\ell}{\fqm}$,
 $\Qm_{22}$ in $\MS{n}{n}{\fq}$ and  $\Qm_{12}$ in $\MS{\ell}{n}{\fqm}$
 with $\norm{\Qm_{12}}=s < t $ so that
\begin{equation}\label{First desc of P in FqN}
\Pm^{-1}=\left( \begin{matrix}
\Qm_{11} & \Qm_{12}\\
\Qm_{21} & \Qm_{22}
\end{matrix} \right).
  \end{equation}
\end{enumerate}
The public key is $(\Gp,\tp)$ with $\tp = t - s$ and
\begin{equation}\label{GPTP1P2EF1}
\Gp=\Sm \left( \Xm \mid \Gm \right) \Pm.
\end{equation}

\paragraph{\bf Decryption.}
We have $\cv \Pm^{-1}=\mv \Sm \left(\Xm \mid \Gm \right) + \ev \Pm^{-1}$.
Suppose that $\ev= \left(\ev_1 \mid \ev_2 \right)$ where $\ev_1 \in \fqm^\ell$ and $\ev_2 \in \fqm^n$. 
We have:
\begin{equation}
\ev \Pm^{-1}=
\left(
\ev_1 \Qm_{11} + \ev_2 \Qm_{21} \mid \ev_1 \Qm_{12} + \ev_2 \Qm_{22}\right)
\end{equation}
It is clear that $\norm{\ev_1 \Qm_{12} + \ev_2 \Qm_{22}} \leq
\norm{\ev_1 \Qm_{12}} + \norm{\ev_2 \Qm_{22}} \leq s + t-s$.
So the plaintext $\mv$ is recovered by applying the decoding
algorithm  only  to the last $n$ components of $\cv\Pm^{-1}$.

We state our main result proving that Overbeck's attack is still successful by considering this time
the dual of $\Lambda _{i}\left(\Gp \right)$ with $i = n - s - k - 1$.

\begin{proposition} \label{prop:crypt:rep}
There exist $\Xm^* \in \MS{k}{\ell +s}{\fqm}$, $\Pm^* \in \GL_{n+ \ell}\left(\F_q \right) $  and
 a generator matrix $\Gm^*$ that defines an $(n-s,k)-$Gabidulin code 
 $\gab{k}{\gv^*}$  such that
\begin{equation} \label{eq:1st}
\Gp= \Sm \left(\Xm^* \mid \Gm^* \right)\Pm^*.
\end{equation}
Furthermore, the error correction capability $t^*$ of $\gab{k}{\gv^*}$ is equal to $t - \frac{1}{2}s$, and hence
$t^* > \tp$.
\end{proposition}
The proof of this  proposition requires to prove the following lemma.

\begin{lemma} \label{lem:decOfP}
There exist
$\Pm_{11}$ in  $\GL_{\ell+s}(\fqm)$, $\Pm_{21}$ in $\MS{(n-s)}{(\ell+s)}{\fqm}$ and $\Pm_{22}$ in 
$\GL_{n-s}(\fq)$ such that 
\begin{equation} \label{eq:generalPm}
\Pm = 
\left( \begin{matrix}
\Im_\ell & ~\ZZ \\
\ZZ & ~\Lm
\end{matrix}
\right)
\left( \begin{matrix}
\Pm_{11} & ~\ZZ\\
\Pm_{21} & ~\Pm_{22}
\end{matrix} \right)
\left( \begin{matrix}
\Im_\ell & ~\ZZ \\
\ZZ & ~\Rm
\end{matrix}
\right)
\end{equation}
with $\Lm$ and $\Rm$ belonging to $\GL_{n}(\fq)$.
\end{lemma}

\begin{proof}
By assumption  $\norm{\Qm_{12}}=s < t$ so there exist $\Rm$ in $\GL_{n}(\fq)$ and $\Qm_{12}^{\prime}$ in $\MS{\ell}{s}{\fqm}$ such that  $\Qm_{12}\Rm = \left(\Qm_{12}^{\prime} \mid \ZZ \right)$.
We set  $\Qm_{22}\Rm = \left(\Qm_{22}^{\prime} \mid \Qm_{23}^{\prime} \right)$
where $\Qm^{\prime}_{22}$ in $\MS{n}{s}{\fq}$  and $\Qm^{\prime}_{23}$ in $\MS{n}{n-s}{\fq}$.
Note that we necessarily have
$\norm{\Qm_{23}^\prime} \le n-s$ and therefore there exists  $\Lm \in \GL_{n}(\fq)$ such that $\Lm
\Qm_{23}^{\prime}=\left( \begin{matrix}
\ZZ\\
\Qm_{23}^{\prime \prime}
\end{matrix}
\right)$ with $\Qm_{23}^{\prime \prime} \in \MS{n-s}{n-s}{\F_q}$. Thus
one can rewrite
\begin{eqnarray}
\left( \begin{matrix}
	\Im_\ell & ~\ZZ \\
	\ZZ & ~\Lm
\end{matrix}
\right)
\Pm^{-1}
\left( \begin{matrix}
\Im_\ell & ~\ZZ \\
\ZZ & ~\Rm
\end{matrix}
\right)
&=&
\left( \begin{matrix}
	\Im_\ell & ~\ZZ \\
	\ZZ & ~\Lm
\end{matrix}
\right) 
\left( \begin{matrix}
\Qm_{11} & \Qm_{12}\\
\Qm_{21} & \Qm_{22}
\end{matrix} \right)
\left( \begin{matrix}
\Im_\ell & ~\ZZ \\
\ZZ & ~\Rm
\end{matrix}
\right)
\\
&=&
\left( \begin{matrix}
\Qm_{11} & ~\Qm_{12}^\prime & ~\ZZ \\
\Lm \Qm_{21} & ~\Lm \Qm_{22}^\prime & ~\Lm \Qm_{23}^\prime
\end{matrix} \right) 
\end{eqnarray}
Observe that there exist
$\Qm_{11}^{\prime \prime}$  in $\MS{\ell+ s}{\ell+s}{\fqm}$ and $\Qm_{21}^{\prime \prime}$ in $\MS{n-s}{\ell+s}{\fqm}$ so that we can write 
\[
\left( \begin{matrix}
	\Im_\ell & ~\ZZ \\
	\ZZ & ~\Lm
\end{matrix}
\right)
\Pm^{-1}
\left( \begin{matrix}
\Im_\ell & ~\ZZ \\
\ZZ & ~\Rm
\end{matrix}
\right)
=
\left( \begin{matrix}
\Qm_{11}^{\prime \prime} & \ZZ\\
\Qm_{21}^{\prime \prime} & \Qm_{23}^{\prime \prime}
\end{matrix} \right).
\]
Note that  $\Qm_{23}^{\prime \prime}$  and  $\Qm_{11}^{\prime \prime}$ are necessarily invertible and thanks to 
Lemma~\ref{InvBlockMat} the proof can be terminated. \qed
\end{proof}

\begin{remark} \label{rk:sandl}
The proof of Lemma \ref{lem:decOfP} is still true if it is assumed that $\norm{\Qm_{12}} < s$, and note that 
by construction $s$ is necessarily less than or equal to $\ell$.
\end{remark}
We are now able to give a proof of Proposition~\ref{prop:crypt:rep}.

\begin{proof}[Proposition~\ref{prop:crypt:rep}]
We keep the same notation as those of Lemma \ref{lem:decOfP}.
Let us rewrite $\Gm \Lm$ as  $\left(\Gm_1^\prime \mid \Gm_2^\prime \right)$ where 
$\Gm_1^\prime$ in $\MS{k}{s}{\fqm}$ and $\Gm_2^\prime$ in $\MS{k}{n-s}{\fqm}$
and set now $\Ym= \left(\Xm \mid \Gm_1^\prime \right)$. Observe that $\Gm_2^\prime$ 
generates an $(n-s, k)-$Gabidulin code.
We then have
\[
\left(\Xm \mid \Gm \right)
\left( \begin{matrix}
\Im_\ell & ~\ZZ \\
\ZZ & ~\Lm
\end{matrix}
\right)
\left( \begin{matrix}
\Pm_{11} & ~\ZZ\\
\Pm_{21} & ~\Pm_{22}
\end{matrix} \right)
= 
\left(\Ym \mid \Gm_2^\prime \right)
\left( \begin{matrix}
\Pm_{11} & ~\ZZ\\
\Pm_{21} & ~\Pm_{22}
\end{matrix} \right)
 = 
 \left(\Xm^* \mid \Gm^* \right) 
\]
 where  $\Xm^* = \Ym \Pm_{11} + \Gm_2^\prime \Pm_{21}$ and 
 $\Gm^* = \Gm_2^\prime \Pm_{22}$ 
is a generator matrix of an $(n-s, k)-$Gabidulin code. 
Hence if we set $\Pm^*=\left( \begin{matrix}
\Im_\ell & ~\ZZ \\
\ZZ & ~\Rm
\end{matrix}
\right)$ we then have rewritten $\Gp$ as expected in \eqref{eq:1st}.
Lastly remark that  $\tp = t - s$  and $t^*=\frac{1}{2} (n-s-k) =\frac{1}{2} (n-k)- \frac{1}{2}s > t-s$.
\qed
\end{proof}

\section{Gabidulin, Rashwan and Honary Variant}

In \cite{GRH09,RGH11} Gabidulin, Rashwan and Honary proposed an other variant 
where the column scrambler has its entries defined on the extension field. 
We will prove that their scheme is actually a special 
case of \cite{G08} and because of that, it suffers the same weakness. So, unlike what it is claimed by the authors, 
 Overbeck's attack is still  successful.

\paragraph{\bf Key generation.}

\begin{enumerate}
\item Pick at random $\gv \in \fqm^n$ such that $\norm{\gv} = n$ and
let $\Gm \in \MS{k}{n}{\fqm}$ be a generator matrix of the Gabidulin code $\gab{k}{\gv}$. Let $\tp$ be an integer $< t$
and set $a \eqdef t - \tp$. 

\item Pick at random $\Sm$ in $\GL_k(\fqm)$ and $\Pm \in \GL_{n}(\fqm)$ such that
\begin{equation}
\Pm^{-1}=\left( \Qm_{1} \mid \Qm_{2}\right)
\end{equation}
where $\Qm_{1} \in \MS{n}{a}{\fqm}$ while $\Qm_{2} \in \MS{n}{n - a}{\fq}$  with
$t=\frac{1}{2}(n-k)$ and $\tp  < t$.
The public key is $(\Gp,\tp)$ with  
\begin{equation}
\Gp=\Sm \Gm \Pm.
\end{equation}
\end{enumerate}

\paragraph{\bf Decryption.}
First, we have $\cv \Pm^{-1}= \mv \Sm  \Gm + \ev \Pm^{-1}$
and $\ev \Pm^{-1}=\left(\ev \Qm_{1}\mid \ev \Qm_{2}\right)$.
Observe that $\norm{\ev \Qm_{1}}\leq a$ and $\norm{\ev \Qm_{2}} \leq \norm \ev \leq \tp$, and since $a = t - \tp$ 
we hence have
 \[
 \norm{\ev \Pm^{-1}} 
 \leq 
 \norm{\ev \Qm_{1}} + \norm{\ev \Qm_{2}}  
\leq t.
\]

We now prove that Overbeck's attack is still successful by considering for this scheme
the dual of $\Lambda _i\left(\Gp \right)$ with $i = n - a - k  - 1$. We first introduce
 the matrices $ \Qm_{11} \in \MS{a}{a}{\fqm}$, $\Qm_{21} \in \MS{n-a}{a}{\fqm}$,
 $\Qm_{12} \in \MS{a}{n-a}{\fq}$ and  $\Qm_{22} \in \MS{n-a}{n-a}{\fq}$ such that
 \begin{equation}
\Pm^{-1}=\left( \begin{matrix}
\Qm_{11} & \Qm_{12}\\
\Qm_{21} & \Qm_{22}
\end{matrix} \right).
  \end{equation}
Note that $\norm{\Qm_{12}} \leq a < t$. Furthermore, by looking at the proof  of 
Lemma~\ref{lem:decOfP}, we can see that this lemma and Proposition~\ref{prop:crypt:rep}
are still true even if $\norm{\Qm_{12}} \leq s$. 
Hence, the scheme given in \cite{GRH09,RGH11} is nothing else but a special 
case of \cite{G08} where $\Xm = \ZZ$ and $\Qm_{12}$ has all its entries in the base field $\fq$.
We have therefore the following corollary.

\begin{corollary}\label{attackSGP}
There exist $\Pm^* \in \GL_{n}(\fq)$ and $\Xm \in \MS{k}{a}{\fqm} $ such that 
\begin{equation}
\Gp=\Sm(\Xm \mid \Gm^*) \Pm^*
\end{equation}
where $\Gm^*$ is a generator matrix of an $(n-a,k)-$Gabidulin code whose
error correction capability $t^*$ is equal to $\lfloor \frac{1}{2} ( t + \tp) \rfloor$, and hence $t^*  > \tp$.
\end{corollary}

\begin{proof}
Apply Proposition~\ref{prop:crypt:rep} with $\ell = 0$ and $s = a$.
Note that the error correction capability $t^*$ of the code $\Gm^*$ is equal to $\frac{1}{2} ( n - a - k)$ that is to say
\[
t^* = t - \frac{1}{2} ( t - \tp) = \frac{1}{2} (t + \tp) > \tp.
\]
\qed
\end{proof}

We summarised in Table~\ref{tab:exp_result} our experimental results obtained with Magma V2.21-6. 
We give the time to find an alternative column scrambler matrix for each parameter 
proposed by the authors in \cite{GRH09} and \cite{RGH11}. In particular, our results outperform those given
in \cite{GRS16,HMR16}.

\begin{table}[h]
\begin{center}
\begin{tabular}{@{~}*{5}{c}@{~}}
\toprule 
$m$ & $k$ & $t$ & $ \tp $ & Time  (second) \\  \midrule
$20$  & $10$ & $5$ & $4$ &  $\leq 1$ \\
$28$ & $14$ & $7$ & $3$  &  $\leq 1$ \\
$28$ & $14$ & $7$ & $4$   & $\leq 1$ \\ 
$28$ & $14$ & $7$ & $5$    & $\leq 1$ \\
$28$ & $14$ & $7$ & $6$  & $\leq 1$ \\
$20$ & $10$ & $5$ & $4$ & $\leq 1$\\
 \bottomrule
\end{tabular}
\end{center}
\caption{Parameters from \cite{GRH09,RGH11} where $n = m$ and at least 80-bit security.} \label{tab:exp_result}
\end{table}

\section{Discussion On a More General Column Scrambler} \label{sec:disc_gen}

In \cite{GRH09} the authors proposed to reinforce the security  by taking a more general column scrambler matrix 
of the form $\Tm \Pm$ where $\Tm$ is  an invertible matrix with its entries in $\fq$ and $\Pm$ is defined over the extension field as it is done in \cite{G08,GRH09,RGH10}. We shall consider Gabidulin's general reparation \cite{G08} since 
\cite{GRH09,RGH10} are particular cases but we emphasize  that this new protection was only defined in \cite{GRH09,RGH10}. 
Assuming that $\Pm$ is then as in \eqref{First desc of P in FqN}, the public key is then of the form
\begin{equation}
\Gp=\Sm \left( \Xm \mid \Gm \right) \Tm  \Pm.
\end{equation}
The decryption of  a ciphertext $\cv$ starts by calculating $\cv \Pm^{-1} \Tm^{-1} =\mv \Sm \left(\Xm \mid \Gm \right) + \ev \Pm^{-1}  \Tm^{-1}$ where $\ev$ is of rank weight $\tp$ and $s = \norm{\Qm_{12}}$. 
The retrieving of the original plaintext  $\mv$ is possible provided 
that $\tp = t -\ell - s$ because  
$\cv \Pm^{-1}\Tm^{-1}=\mv \Sm \left(\Xm \mid \Gm \right) + \ev \Pm^{-1}\Tm^{-1}$.
Suppose that $\ev= \left(\ev_1 \mid \ev_2 \right)$ where $\ev_1 \in \fqm^\ell$ and $\ev_2 \in \fqm^n$,
then we also have
\begin{equation}
\ev \Pm^{-1} \Tm^{-1}=
\left(
\ev_1 \Qm_{11} + \ev_2 \Qm_{21} \mid \ev_1 \Qm_{12} + \ev_2 \Qm_{22}\right)\Tm^{-1}.
\end{equation}
It is clear that $\norm{\ev_1 \Qm_{12} + \ev_2 \Qm_{22}} \leq
\norm{\ev_1 \Qm_{12}} + \norm{\ev_2 \Qm_{22}} \leq s + \tp$ and hence it implies that
\[
\norm{\ev \Pm^{-1} \Tm^{-1}} = \norm{\ev \Pm^{-1}} \leq \norm{ \ev_1 \Qm_{11} + \ev_2 \Qm_{21} } + \norm{\ev_1 \Qm_{12} + \ev_2 \Qm_{22}} \leq \ell + s + \tp.
\]
Therefore the plaintext $\mv$ is recovered by applying the decoding
algorithm  only  to the last $n$ components of $\cv\Pm^{-1}\Tm^{-1}$. But in this case, the rank weight of the last $n$ components of $\ev \Pm^{-1} \Tm^{-1}$ is not  necessarily less than or equal to  $\tp+s$ but rather to $\tp+s+\ell$. Consequently, the decryption  will always succeed if it assumed that $\tp= t - s - \ell$ otherwise the decoding may fail.
Hence, we see why this new reparation was just proposed for the case where $\ell = 0$ \textit{i.e.} 
without any distorsion matrix since otherwise its deteriorates the performances of the original scheme.

We now study more precisely the  security 
this protection might bring in for  the general scheme of \cite{G08}. 
First, rewrite $\Tm$ as
\begin{equation}
\Tm =\left( \begin{matrix}
\Tm_{11} & \Tm_{12}\\
\Tm_{21} & \Tm_{22}
\end{matrix} \right)
  \end{equation}
where  $ \Tm_{11} \in \MS{\ell}{\ell}{\fq}$, $ \Tm_{21} \in \MS{n}{\ell}{\fq}$,
$ \Tm_{12} \in \MS{\ell}{n}{\fq}$
and $ \Tm_{22} \in \MS{n}{n}{\fq}$. 
On the other hand, by Lemma~\ref{lem:decOfP} the matrix $\Pm$ can be expressed as \eqref{eq:generalPm}.
We can find then $\Xm_{1}$ in $\MS{k}{(\ell+s)}{\fqm}$ and $\Xm_{2}$ in 
$\MS{k}{(n - s)}{\fqm}$ such that
\begin{eqnarray*}
 \left( \Xm \mid \Gm \right) \Tm\Pm 
 &=&  \left ( \Xm_1 \mid \Xm_2 + \Gm_1^* \right)
\end{eqnarray*}
where $\Gm_1^*$ generates an $(n-s,k)-$Gabidulin code and $\norm{\Xm_2} = \norm{\Xm} \leq \ell$.
From Proposition~\ref{prop:XtoX2=0} and by taking $t_2 = \ell$, 
there exist
 $\Pm^* \in \GL_{n+\ell}(\fq)$, $\Xm^* \in \MS{k}{(2\ell+s)}{\fqm}$ 
and $\Gm^*$ that generates an $(n - s - \ell)-$Gabidulin code 
such that
\begin{equation} 
 \left( \Xm \mid \Gm \right) \Tm \Pm = \left (\Xm^*  \mid \Gm^* \right) \Pm^*.
\end{equation}
We have therefore proven the following proposition.

\begin{proposition} \label{prop:TP}
Assume that $\Gp=\Sm \left( \Xm \mid \Gm \right) \Tm \Pm$ where  
$\Tm \in \GL_{n+\ell}(\fq)$ and $\Pm$ has the form \eqref{First desc of P in FqN}.
There exist then $\Pm^*$ in $\GL_{n+\ell}(\fq)$, $\Xm^*$ in $\MS{k}{(2\ell + s)}{\fqm} $ and 
a matrix $\Gm^*$ that generates an $(n - s - \ell,k)-$Gabidulin code 
 $\gab{k}{\gv^*}$   such that 
\[
\Gp
=
\Sm (\Xm^* \mid \Gm^*)\Pm^*. 
\]
Furthermore, the correction capability $t^*$ of  $\gab{k}{\gv^*}$ is greater than $t - \frac{1}{2}(\ell +s)$.  
In particular $t^* > \tp$.
\end{proposition}

This result shows that actually this new proposed protection does not improve the security even when
applied with the scheme for which a distortion matrix $\Xm$ is used. 
 An example where this protection was used and turns out to be useless is  the scheme 
 given in \cite{GRH09,RGH10}.
 
\section{Conclusion}

The apparition of Overbeck's attack prompted some authors to invent reparations to hide more the structure of the Gabidulin codes. 
One trend advocated the use of a right column scrambler with entries in the extension field as it is done in 
 \cite{G08,GRH09,RGH11}.
Our analysis shows that these reparations aiming at resisting Overbeck's structural attack 
do fail precisely against it. By applying appropriately Overbeck's technique, 
 we were able to construct a Gabidulin code that has the same dimension as the original one but with a lower length. 
 Hence,  we obtain a degraded Gabidulin code in terms 
 of error correction capabilities but we prove  that the degradation does not forbid the error correction of any ciphertext. 
 Furthermore, when the attack is implemented, the  practical results we obtained outperform those given
in \cite{GRS16,HMR16} which were up to our paper 
the best attacks against the schemes of \cite{G08,GRH09,RGH11}.
We also considered in Section~\ref{sec:disc_gen} the case where an isometric transformation is applied  in conjunction  with a right column scrambler 
which has its entries in the extension field. We proved that this protection is useless both in terms of performance and security.

\medskip

The other kind of reparation is followed by the series of works in \cite{L10,RGH10} which propose to resist to Overbeck's attack by taking a distortion matrix 
$\Xm$ so that the codimension of $\Lambda _{n-k-1}\left(\CA \right)$ is equal to $a$ where $a$ is sufficiently large to prevent an exhaustive search. 
But these reparations were cryptanalyzed in \cite{GRS16,HMR15}.

\medskip

Furthermore, since the attack in \cite{HMR15} only considers column scrambler matrices on the base field, one may try to avoid it by combining 
the reparations proposed in  \cite{L10,RGH10} with those of \cite{G08,GRH09,RGH11}. 
Nevertheless, our paper shows that the security of \cite{G08,GRH09,RGH11} can be reduced to the one with a column scrambler with entries in the base field. 
Consequently, using our results and then applying the general attack of \cite{HMR15}
may break this ``patched'' scheme.

\appendix
\newpage

\section{Overbeck's Attack} \label{appendix:overb_attack}

Let assume that $\Gp=\Sm\left( \Xm \mid \Gm\right) \Pm$ is 
the public generator matrixthat generates $\Cp$
with $\Pm \in \GL_{n+\ell}(\fq)$, $\Xm \in \MS{k}{\ell}{\fqm}$ and $\Gm$ generates  a Gabidulin code $\gab{k}{\gv}$ where $\norm{\gv} = n$. Observe that  $\Lambda _{i}(\Gp)$ can be written as 
\begin{eqnarray}
\Lambda _{i}\left(\Gp \right)
& = &
\Sm_{\rm ext}%
\Big ( \begin{matrix}
\Lambda _{i}\left(\Xm\right) \mid \Lambda _{i}\left(\Gm\right)%
\end{matrix}%
\Big )
\Pm
~~\text{ where } ~
\Sm_{\rm ext}\eqdef 
\begin{pmatrix} 
\Sm^{[0]} &               &  \ZZ\\
               & \ddots{}  &      \\
    \ZZ            &               &  \Sm^{[i]}
\end{pmatrix}.
\end{eqnarray}
Since $\Lambda _{i}\left(\Gm\right)$ generates $\gab{k+i}{\gv}= \gab{n-1}{\gv}$, there exists $\Sm^\prime \in \GL_{k(i+1)i}(\fqm)$ such that 
\begin{equation}\label{reduction1}
\Sm^\prime \Lambda _{i}\left(\Gp \right) = 
\left(
\begin{matrix}
 \Xm^* & \Gm_{n-1} \\
\Xm^{**} & \ZZ 
\end{matrix}
\right)
\Pm
\end{equation}
where $\Xm^* \in \MS{(n-1)}{\ell}{\fqm}$, $\Xm^{**} \in \MS{(k(i+1)-n+1)}{\ell}{\fqm}$ and $\Gm_{n-1} \in \MS{(n-1)}{n}{\fqm}$ generates $\gab{n-1}{\gv}$. Using \eqref{reduction1}, one can deduce that  
by taking $ i = n - k - 1$
\[
\dim{\Lambda _{n - k - 1}(\Cp)} = n-1+ \rank (\Xm^{**}).
\]
In the particular case where $\rank (\Xm^{**})= \ell$ then $\dim{\Lambda _{i}(\Cp)}=n+\ell-1$ and 
 thus $\dim{\Lambda _{i}(\Cp)^\perp} = 1$. Furthermore, if 
$\hv$ is a non zero vector from $\dual{\gab{n-1}{\gv}}$ and we set 
$\hv^*=\left(\ZZ \mid \hv \right)\left( \Pm^{-1} \right)^T$
then under the assumption that $\rank (\Xm^{**})= \ell$ we have
\begin{equation} \label{eq:dualOver}
\Lambda_{n - k - 1}(\Cp)^{\perp}= \fqm  \hv^*.
\end{equation}

\begin{proposition}
Let $\vv \in \Lambda_{n-k-1}(\Cp)^{\perp}$ with $\vv \ne \ZZ $. 
Any matrix $\Tm \in \GL_{n+ \ell}(\fq)$ that satisfies $\vv\Tm =\left(\ZZ \mid \hv^{\prime }\right)$ 
with $\hv^\prime \in \fqm^n $ is an alternative column scrambler matrix, that is to say, 
there exist $\Zm$ in $\MS{k}{\ell}{\fqm}$ and $\Gm^*$ that generates  a Gabidulin code $\gab{k}{\gv^*}$ such that
\[
\Gp=\Sm\left( \Zm \mid \Gm^* \right) \Tm.
\]
\end{proposition}

\begin{proof}
From \eqref{eq:dualOver} there exists $\alpha \in \fqm$ such that $\vv= \alpha \hv^* = (\ZZ \mid \alpha \hv)\left( \Pm^{-1} \right)^T$ where $\hv$ is a non zero vector of $\dual{\gab{n-1}{\gv}}$.
Let $\Tm \in \GL_{n+\ell}(\fq)$ such that $\vv\Tm^{T}=(\ZZ \mid \hv^{\prime })$ and
consider the matrices  $\Am \in \MS{\ell}{\ell}{\fq}$  and $\Dm \in \MS{n}{n}{\fq}$ so that
\[
\Tm \Pm^{-1} =
\left(
\begin{matrix}
\Am & \Bm \\
\Cm & \Dm
\end{matrix}
\right).
\]
We have then the following equalities
\begin{equation}\label{reduction2}
\tilde{\hv} \Tm^T = \left(\ZZ \mid \alpha \hv\right)\left(\Pm^{-1}\right)^T \Tm^T =\left(\ZZ \mid \alpha \hv\right)\left(\Tm \Pm^{-1}\right)^T  = (\ZZ \mid \hv^{\prime})
\end{equation}
 It comes out from \eqref{reduction2} that $\hv \Bm^T=\ZZ $ and hence $\Bm = \ZZ$ since $\norm \hv = n$. So we can write $\Tm \Pm^{-1} =
\left(
\begin{matrix}
\Am & \ZZ \\
\Cm & \Dm
\end{matrix}
\right)
$
and using Lemma~\ref{InvBlockMat}, $\Pm \Tm^{-1} =
\left(
\begin{matrix}
\Am^\prime & \ZZ \\
\Cm^\prime & \Dm^\prime
\end{matrix}
\right)
$.
Consequently,
\[
\Gm_{pub}\Tm^{-1}=\Sm\left(\Xm \mid \Gm \right)\left(
\begin{matrix}
\Am^\prime & \ZZ \\
\Cm^\prime & \Dm^\prime
\end{matrix}
\right)
=\Sm\left(\Zm \mid \Gm^* \right)
\]
where $\Gm^* = \Gm \Dm^\prime$ is a generator matrix of an $\left(n,k \right)-$Gabidulin code. So $\Tm$ is an alternative column scrambler matrix for the system. 
\end{proof}

\bibliographystyle{alpha}
\bibliography{codecrypto2}
\end{document}